\documentclass[journal]{IEEEtran}
\usepackage{bbm}
\usepackage{graphicx,subfigure}
\usepackage{mathrsfs}
\usepackage{algorithm, algpseudocode}
\usepackage{amssymb}
\usepackage{amsmath}
\usepackage{amsfonts}
\usepackage{amsthm}
\usepackage{psfrag}
\usepackage{epsfig}
\usepackage{graphicx,subfigure}
\usepackage{booktabs}
\usepackage{color}
\usepackage{xcolor}
\usepackage{colortbl}
\usepackage{enumerate}
\usepackage{cite}
\usepackage{cases}
\usepackage[numbers, sort]{natbib}
\usepackage[caption=true]{caption}

\newtheorem{theorem}{Theorem}
\newtheorem{lemma}{Lemma}
\newtheorem{definition}{Definition}
\newtheorem{remark}{Remark}

\newtheorem{assumption}{Assumption}

\allowdisplaybreaks
\begin{document}
\title{Self-Triggered Distributed Model Predictive Control with Synchronization Parameter Interaction}

\author{Qianqian Chen}

\maketitle

\begin{abstract}
	This paper investigates an aperiodic distributed model predictive control approach for multi-agent systems (MASs) in which  parameterized synchronization constraints is considered and  an innovative self-triggered criterion is constructed. Different from existing coordination methodology, the proposed strategy achieves the cooperation of agents through the synchronization of one-dimensional parameters related to the control inputs. At each asynchronous sampling instant, each agent exchanges the one-dimensional synchronization parameters, solves the optimal control problem (OCP) and then determines the open-loop phase. The incorporation of the self-triggered scheme and the synchronization parameter constraints relieves the computational and communication usage. Sufficient conditions guaranteeing the recursive feasibility of the OCP and the stability of the closed-loop system are proven. Simulation results illustrate the validity of the proposed control algorithm.
\end{abstract}


\section{Introduction}

\IEEEPARstart{T}{he} thriving expansion of modern industrial sectors has lead to the application of MASs in various complex fields, such as aerospace, smart grid \cite{yang2021distributed}, and machine learning \cite{yin2020deep}. To efficiently accomplish the cooperative tasks, different control algorithms have emerged as powerful platforms for MASs. Distributed model predictive control (DMPC) \cite{christofides2013distributed, zhou2022multiple}  has been developed to address the challenges posed by strong nonlinearity, multi variables and complex uncertainties in controlled systems. This approach allows the closed-loop system to meet physical, security and economic constraints while achieving optimization performance. DMPC has been successfully applied to coordination tasks such as consensus \cite{wang2023distributed}, formation control \cite{yuan2023distributed}, trajectory tracking \cite{hao2023trajectory}, and reconfigurable control \cite{hou2021distributed, zheng2023distributed}, resulting in numerous interesting outcomes.

The traditional predictive control often requires solving the OCP repeatedly at each time step, which is computationally intensive \cite{yang2021economic}. This limitation becomes apparent when applying control strategies to the systems that have requirement of rapidity. To address this issue, event-triggered DMPC and self-triggered DMPC have emerged as popular approaches for MASs that balance resource utilization and operation performance. For example, in \cite{mi2019self}, a dynamic programming based co-design scheme is proposed for linear MASs which integrates the self-triggering scheme into DMPC. In \cite{chen2023asynchronous}, a stochastic self-triggered DMPC problem is addressed for vehicle platoon MASs under vehicular ad-hoc networks. In \cite{wei2021self}, a self-triggered min-max DMPC algorithm is proposed to achieve formation stabilization for nonlinear MASs with communication delays. In \cite{wang2024rolling}, a rolling self-triggered DMPC scheme is proposed to mitigate the mutual influences caused by dynamic couplings. The key difference between event-triggered DMPC and self-triggered DMPC lies in the approach to determine the triggering moments. Event-triggered DMPC necessitates continuous monitoring of the system's state at each instant, whereas self-triggered DMPC only requires state measurement at the trigger instants. Compared to event-triggered DMPC, self-triggered DMPC requires less state information and offers greater autonomy. This strategy allows agents to interact at specific instants determined by themselves, rather than continuously.

In the above results, the \textit{full-state}  information is interacted among neighboring agents. When the dimension of system state is substantial and communication bandwidth is limited, the full-state interaction based algorithm may be unimplementable. The ``Divide-to-Conquer'' based cooperation strategy is proposed in  \cite{skjetne2004robust, zhang2012distributed} to divide the cooperative of MASs into a geometric and dynamic task. Under this strategy, a synchronization parameter is predefined according to the specific task form. The geometric task requires that each agent follows the reference agent on the desired parameterized state trajectory. The dynamic task ensures that the dynamic behavior of reference agent can be satisfied.
When the parameters to be synchronized of reference agent converge to be identical, the desired formation behavior is consequently achieved. The advantage of this mechanism lies in its ability to decouple the path tracking and inter-vehicle coordination. The interaction of the reference agents is through only one-dimensional synchronization parameters rather than heavy full-state information. As a result, the cooperation of MASs is achieved by a coupled synchronization parameters constraint added to the OCP constructed in DMPC approach. Existing research \cite{qin2023event, qin2023asynchronous} recognize the critical role played by the parameterized constraint in the event-triggered DMPC. However, few scholars have been able to draw on any systematic research into the self-triggered DMPC involving the synchronization parameter constraints. Since the interaction based on the mere synchronization parameters between the MASs, some challenges have hindered the development of this control strategy. For instance: 1) how to construct the estimated sequence of one-dimensional synchronization parameter for neighbor agents under the event-based communication mechanism; 2) how to design preferred self-triggering conditions independent of precise state information of neighbor agents. Hence, this study provides new insights into these issues.

This paper investigates a one-dimensional synchronization parameter interaction based self-triggered DMPC algorithm for nonlinear MASs formation tracking. The overall goal of this research is to pursue the cooperation of MASs subject to bounded disturbances through less interactive information and solving frequency. The main contributions are summarized as follows. (1) Compared to the existing work, the coordination of agents is facilitated through synchronization parameter interaction by effectively estimating the relevant neighbor information.  (2) A composite self-triggered mechanism merely based on explicit self-information is proposed while an extra stability-based constraint is introduced. As a result, the transmission between the subsystem and controller is intermittent and autonomous. (3) Sufficient criteria to guarantee the recursive feasibility and stability are developed, which form the theoretical basis for the algorithm's implementation.

\emph{Notations:} Let $\mathbb{R}$ and $\mathbb{N}$ denote the reals and nonnegative integers, respectively. $\mathbb{R}_{+}$ and $\mathbb{N}_{+}$ denote the positive reals and positive integers. $\mathbb{N}_{[a,b]}=\left\{c\in \mathbb N\mid a\leq c\leq b\right\}$. For a vector $x \in \mathbb{R}^n$, $\|x\|=\sqrt{x^Tx}$ is the Euclidean norm and $\|x\|_P=\sqrt{x^TPx}$ is the $P$-weighted norm. For a matrix $A$, $A>0$ means $A$ is positive-definite symmetric. For a set $S$ and an event $E \subseteq S$, let an indicator function $\mathcal{I}\{E\}=1$ (respectively, $\mathcal{I}\{E\}=0$) if event $E$ is valid (respectively, event $E$ is not valid).

\section{Preliminaries and Problem Statements}
\subsection{Problem Formulation}
Consider a class of nonlinear MASs composed of $M$ error subsystems for tracking, where subsystem $i \in \mathcal{M} \triangleq \mathbb{N}_{[1,M]}$ is described as
\begin{eqnarray}\label{subsystem model}
    e_i(k+1) = f_i(e_i(k), u_i(k)) + d_i(k),
\end{eqnarray}
where $e_i(k) \in \mathbb{R}^{n_i}$, $u_i(k) \in \mathbb{R}^{m_i}$ and $d_i(k) \in \mathbb{R}^{n_i}$ are the error state, the control input and the external input, respectively. Assume the subsystem is subject to
\begin{eqnarray*}
    e_i(k) \in \mathbb{E}_i, \quad u_i(k) \in \mathbb{U}_i, \quad d_i(k) \in \mathbb{D}_i,
\end{eqnarray*}
where $\mathbb{E}_i \subset \mathbb{R}^{n_i}$, $\mathbb{U}_i \subset \mathbb{R}^{m_i}$, $\mathbb{D}_i = \{ d_i \in \mathbb{R}^{n_i} : \left\| d_i \right\| \leq \eta_i, \eta_i \in \mathbb{R}_{+} \}$. The linearized system of (\ref{subsystem model}) at origin is derived as $e_i(k+1) = G_i e_i(k) + H_i u_i(k)$, where $G_i = \partial f_i / \partial e_i(0,0)$, $H_i = \partial f_i / \partial u_i(0,0)$. Assume there exists matrix $K_i$ such that $G^c_i = G_i + H_i K_i$ is schur stable, and the following assumption is satisfied:
\begin{assumption}\cite{sun2019robust}\label{ass: Lipschitz}
   The nonlinear function $f_i(e_i, u_i)$ satisfying $f_i(\mathbf{0}, \mathbf{0}) = \mathbf{0}$ is locally Lipschitz continuous with respect to $e_i$. The function $g_i(e_i, u_i) = f_i(e_i, u_i) - e_i$  is $L_{g_i}$-locally Lipschitz continuous with respect to $e_i$. Furthermore, for $e_i \in \Omega_{r_i}$ ($\Omega_{r_i}$ is a terminal region defined in Assumption \ref{ass: terminal}), the function $g_i(e_i, \kappa_i(e_i))$ with local controller $\kappa_i(e_i)$  is $L_{\kappa_i}$-locally Lipschitz continuous with respect to $e_i$.
\end{assumption}

Before we present the concrete control strategy, the following definition is introduced.
\begin{definition}\label{def: ISS}
	Each subsystem in (\ref{subsystem model}) is said to admit  an ISS-Lyapunov function $V_i(e_i(k)):\mathbb E_i\mapsto \mathbb R_+$ if the following conditions hold:
	\begin{itemize}
		\item[1)] $\underline \alpha_i(\|e_i(k)\|) \leq V_i(e_i(k)) \leq \bar \alpha_i(\|e_i(k)\|)$, where $\underline \alpha_i(\cdot)$ and $\bar \alpha_i(\cdot)$ are $\mathcal{K}_{\infty}$ functions.
	\end{itemize}
	\begin{itemize}
		\item[2)] $ V_i(e_i(k+1)) - V_i(e_i(k)) \leq -\alpha_i(\|e_i(k)\|) + \delta_i(\|d_i(k)\|)$, where $\alpha_i(\cdot)$ is a $\mathcal{K}_{\infty}$ function, $\delta_i(\cdot)$ is a $\mathcal{K}$ function.
	\end{itemize}
\end{definition}

Our purpose is to design a self-triggered coordination algorithm to reduce unnecessary communication and energy consumption.  Let $\mathcal{N}_i$ denotes the neighbor set of agent $i$, then agent $i$ can receive information from agent $j$ for any $j \in \mathcal{N}_i$. Motivated by \cite{qin2023event, qin2023asynchronous}, each agent executes two assignments for formation:
\begin{enumerate}
	\item
	Geometric task: Agent $i$ is forced to converge to and follow a desired parameterized path $\Gamma_i = \{ \chi_i(k) \in \mathbb{R}^{n_i} : \exists s_i(k) \in \mathbb{R}, \;\; \text{s.t.} \;\; \chi_i(k) = d_{ri}(s_i(k)) \}$, where $\chi_i(k)$ is the position of agent $i$, $s_i(k)$ is the synchronization parameter, $d_{ri}(s_i(k))$ is the desired position. In other words, $\lim_{k \rightarrow \infty} \| \chi_i(k) - d_{ri}(s_i(k)) \| = 0$ is required.
	
	\item
	Dynamic task: Agent $i$ is forced to dynamically coordinate with  neighboring agent $j$ by minimizing the difference of synchronization parameter between $s_i(k)$ and $s_j(k)$ such that their behaviors evolve similarly. That is, $\lim_{k \rightarrow \infty} | s_i(k) - s_j(k) | = 0$ is required.
\end{enumerate}

The implementation of dynamic task is to achieve the cooperation of reference agents. Then, a tracking controller is designed to make agent $i$ converge towards the corresponding reference agent.  The overall control schematic is shown in Fig. \ref{Fig: schematic}. To realize coordination in a manner of rolling optimization, we specify the updating law of the synchronization parameter $s_i(k)$ as
\begin{eqnarray}\label{syn model}
	s_i(k+1) = h_i(s_i(k), u_i(k),  u_{ri}(k)),
\end{eqnarray}
where $h_i$ takes arguments in  synchronization parameter $s_i(k)$, control input $u_i(k)$ and reference input $u_{ri}(k)$. Without loss of generality, $h_i = s_i(k) + T\cdot[Y_i u_i(k) + Z_i u_{ri}(k)]$, where $T$ is the preset sampling period,  $Y_i=[y_{1}, ...,y_{m_i}]\in \mathbb{R}^{1 \times m_i}$ and $Z_i=[z_{1}, ...,z_{m_i}] \in \mathbb{R}^{1 \times m_i}$ are optional weight parameters.

\begin{figure*}
	\centering
	\includegraphics[width=18cm]{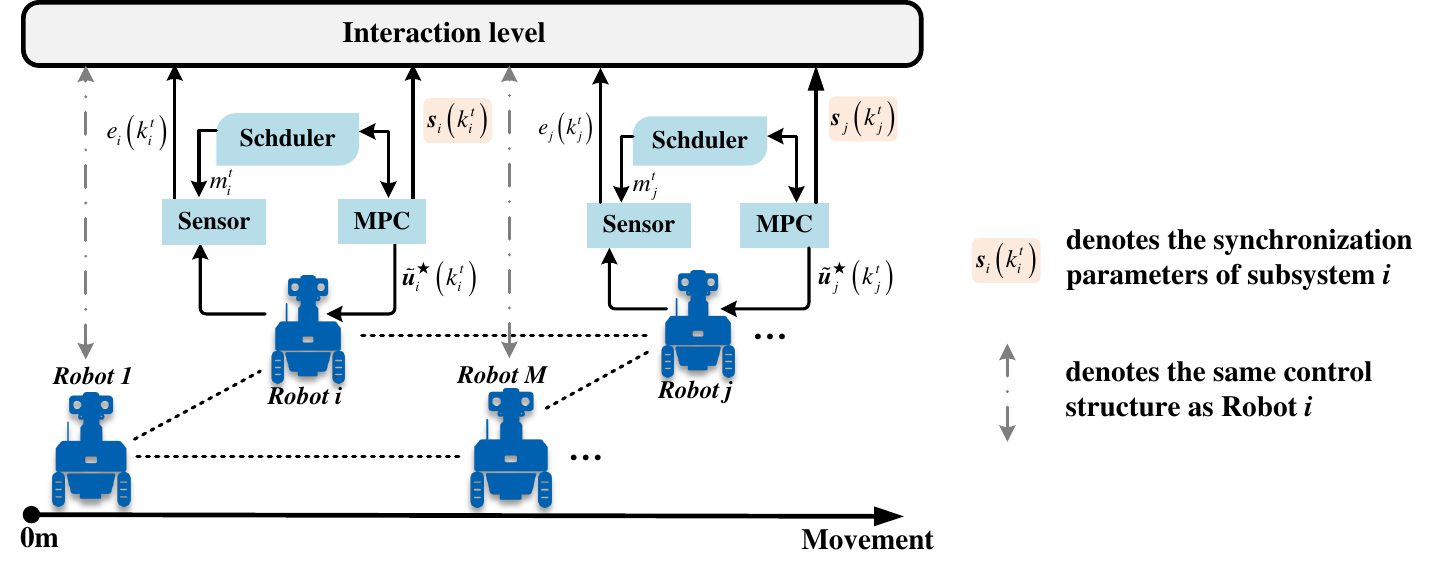}
	\caption{The Schematic of the proposed self-triggered DMPC strategy with $M$ agents.}
	\label{Fig: schematic}
\end{figure*}

\subsection{Optimal Control Problem}
The nominal model of subsystem $i$ is expressed as
\begin{eqnarray}\label{subsystem nominal model}
	\tilde e_i(k+1) = f_i(\tilde e_i(k), \tilde u_i(k)),
\end{eqnarray}
where $\tilde e_i(k) \in \mathbb{R}^{n_i}$ and $\tilde u_i(k) \in \mathbb{R}^{m_i}$. Let $N \in \mathbb{N}_+$ be the prediction horizon, then at event $(e_i, k)$ (i.e., for state $e_i$ at instant $k$), the cost function to be minimized is defined as
\begin{align*}
	&J_i(e_i(k), \tilde{\textbf{u}}_i(k), \textbf{s}_j(k)) \\
	&= H_i(e_i(k), \tilde{\textbf{u}}_i(k)) + \sum^{N-1}_{\tau = 0} \sum_{j\in\mathcal{N}_i} \rho_{ij} |s_i(\tau|k) - s_j(\tau|k)|^2,
\end{align*}
where
\begin{align*}
	H_i(e_i(k), \tilde{\textbf{u}}_i(k))
	&= \sum^{N-1}_{\tau = 0}(\| \tilde e_i(\tau|k)\|^2_{Q_i} + \| \tilde u_i(\tau|k)\|^2_{R_i}) \\
	&\quad + \| \tilde e_i(N|k)\|^2_{P_i},
\end{align*}
where $*(\tau|k)$ represents the predicted value of $*(k+\tau)$ at instant $k$ obtained according to the prediction model (\ref{subsystem nominal model}) for $\tau \in \mathbb{N}_{[0,N]}$. $\tilde{\textbf{u}}_i(k)= \{\tilde u_i(0|k),...,\tilde u_i(N-1|k)\}$ denotes the control sequence to be solved; $\textbf{s}_j(k) = \{s_j(0|k),...,s_j(N-1|k)\}$ denotes the sequence of synchronization parameters of agent $j$; $H_i(e_i(k), \tilde{\textbf{u}}_i(k))$ is the stage cost constructed by the current actual state $e_i(k)$ and the predicted control sequence $\tilde{\textbf{u}}_i(k)$. To guarantee the stability, matrices $Q_i >0$, $R_i > 0$ and $P_i > 0$ meet the following mild assumption:
\begin{assumption}\label{ass: terminal}
	For nominal system (\ref{subsystem nominal model}) and given matrices $Q_i >0$, $R_i > 0$, there exists a constant $\varepsilon_{ri} \in \mathbb{R}_+$, a matrix $P_i>0$, such that for $\forall \tilde e_i(k) \in \Omega_i(\varepsilon_{ri}) \triangleq \{ \tilde e_i : \| \tilde e_i(k) \|^2_{P_i} \leq \varepsilon^2_{ri}\} $, the following formulae hold: 1) $K_i \tilde e_i(k) \in \mathbb{U}_i$, $\tilde e_i(k+1) \in \Omega_i(\varepsilon_{ri})$; 2) $\| \tilde e_i(k+1) \|^2_{P_i} - \| \tilde e_i(k) \|^2_{P_i} \leq -\| \tilde e_i(k) \|^2_{Q^*_i}$, where $Q^*_i = Q_i + K^T_i R_i K_i$.
\end{assumption}

In the framework of self-triggered mechanism, each agent independently analyzes environment information and decide the open-loop phase, after which to communicate or take action as needed. Denote the $t$-th ($t \in \mathbb{N}_+$) triggering instant of agent $i$ as $k^t_i$. The optimal control problem to be solved at instant $k^t_i$ is described as
\begin{align}
	& \textbf{OCP$_i:$} \min_{\tilde u_i(\tau|k^t_i)} J_i(e_i(k^t_i), \tilde{\textbf{u}}_i(k^t_i), \hat{\textbf{s}}_{i,j}(k^t_i)) \label{OCP: cost function}\\
	& = H_i(e_i(k^t_i), \tilde{\textbf{u}}_i(k^t_i)) \nonumber \\
	&\quad + \sum^{N-1}_{\tau = 0} \sum_{j\in\mathcal{N}_i} \rho_{ij} |s_i(\tau|k^t_i) - \hat s_{i,j}(\tau|k^t_i)|^2, \nonumber \\
	& \text{subject to for} \;\; \tau \in \mathbb{N}_{[0,N-1]}  \nonumber \\
	& \tilde e_i(0|k^t_i) = e_i(k^t_i), \nonumber\\
	& \tilde e_i(\tau+1|k^t_i) = f_i(\tilde e_i(\tau|k^t_i), \tilde u_i(\tau|k^t_i)), \nonumber\\
	& s_i(\tau+1|k^t_i) = h_i(s_i(\tau|k^t_i), \tilde u_i(\tau|k^t_i), u_{ri}(k^t_i+\tau)), \nonumber\\
	& \tilde u_i(\tau|k^t_i) \in \mathbb{U}_i, \label{OCP: control}\\
	& \tilde e_i(\tau|k^t_i) \in \mathbb{E}_i \ominus \Lambda_i(\tau),  \label{OCP: state}\\
	& \tilde e_i(N|k^t_i) \in \Omega_i(\varepsilon_i),  \label{OCP: terminal}\\
	& H_i(e_i(k^t_i), \tilde{\textbf{u}}_i(k^t_i)) \leq \phi_i(k^t_i), \label{OCP: stability}
\end{align}
where $\Lambda_i(\tau) \triangleq \{ \tilde e_i: \|\tilde e_i\|_{P_i} \leq \tau \eta_i \lambda_{\text{max}} (\sqrt{P_i}) (1+L_{gi})^{\tau-1}  \}$ is the tightened constraint set as in \cite{xie2021robust}, and
\begin{align}
	\phi_i(k^t_i)
	& = - \|e_i(k^t_i - 1)\|^2_{Q_i} - \|\tilde u^{\star}_i(m^{t-1}_i - 1 | k^{t-1}_i)\|^2_{R_i} \nonumber\\
	&\quad + H_i(e_i(k^t_i-1), \bar{\textbf{u}}_i(k^t_i-1)) + \psi_i(m^{t-1}_i) \eta_i, \label{eqn: phi}
\end{align}
where $\psi_i(m^{t-1}_i)$ will be defined in Theorem \ref{theo: trigger}, and the item $H_i(e_i(k^t_i-1), \bar{\textbf{u}}_i(k^t_i-1))$ will be replaced by $H_i(e_i(k^t_i-1), \textbf{u}^{\star}_i(k^t_i-1))$ when $m^{t-1}_i=1$. Since the interaction between the agents are asynchronous, at instant $k^t_i$, agent $i$ receive the synchronization parameter information transmitted by neighbor agent $j$ at the closest moment $k^t_{i,j} \triangleq \arg \min_{k^t_{j}} k^t_{i} - k^t_{j}, \; \text{s.t.} \; k^t_{j} < k^t_{i} $. The assumed sequence $\hat{\textbf{s}}_{i,j}(k^t_i)$ introduced in (\ref{OCP: cost function}) is established based on the latest received optimal solution $\tilde{\textbf{u}}^{\star}_j(k^t_{i,j})$ at instant $k^t_{i,j}$ from neighbor $j$:
\begin{eqnarray*}
	\hat s_{i,j}(\tau|k^t_i) =
	\begin{cases}
		s^{\star}_{j}(k^t_i+\tau|k^t_{i,j}) & \text{if}\;\; \tau \in \mathbb{N}_{[0,N-k^t_i-1]}, \\
		\mu_{j} \hat s_{i,j}(\tau-1|k^t_i) & \text{if}\;\; \tau \in \mathbb{N}_{[N-k^t_i,N]},
	\end{cases}
\end{eqnarray*}
where $\mu_j \in \mathbb{R}$ is the estimated parameter based on synchronous parameters updating law (\ref{syn model}).

\begin{remark}
	The item $\sum_{j\in\mathcal{N}_i} \rho_{ij} |s_i(\tau|k^t_i) - \hat s_{i,j}(\tau|k^t_i)|^2$ in (\ref{OCP: cost function}) is called the coupled parameterized synchronization  constraint. Motivated by \cite{qin2023asynchronous, qin2023event}, the introduction of an extra variable $s_i$ can describe the information of the reference agents along the parameterized state trajectory. By constructing an error system that enables agents to track a reference agent, the desired formation can be achieved as long as the reference agent achieves synchronization. Hence, the proposed strategy renders that each agent exchanges one-dimensional parameters rather than heavy full-state to cooperate.
\end{remark}
\section{Self-triggering mechanism}
In this section, we propose the self-triggering autonomous criteria to determine the open-loop phase for each agent $i\in\mathcal{M}$. Denote the successive triggering instants as $k^t_i$ and $k^{t+1}_i$. Based on the optimal solution of the prior instant $\tilde u^{\star}_i(\tau|k^t_i)$, $\tau\in\mathbb{N}_{[0,N]}$, we can construct the candidate control input $\bar u_i(\tau|k^{t+1}_i)$ for subsystem $i$ as follows:
\begin{eqnarray}\label{eqn: feasible solution}
	\bar u_i(\tau|k^{t+1}_i) =
	\begin{cases}
		\tilde u^{\star}_i(\tau+m^t_i|k^t_i),
		 & \text{if}\;\; \tau \in \mathbb{N}_{[0, N-m^t_i-1]}, \\
		K_i \bar u_i(\tau|k^{t+1}_i),
		 & \text{if}\;\; \tau \in \mathbb{N}_{N-m^t_i, N]}.
	\end{cases}
\end{eqnarray}

The corresponding candidate state sequence can be obtained through $\bar e_i(\tau+1|k^{t+1}_i) = f_i(\bar e_i(\tau|k^{t+1}_i), \bar u_i(\tau|k^{t+1}_i))$ with initial condition $\bar e_i(0|k^{t+1}_i) = e_i(k^{t+1}_i)$. The self-triggering conditions guaranteeing stability will be derived in this section, where the essential spirit lies in making the Lyapunov function decrease. Before we present the self-triggering requisites, the following lemma is given.
\begin{lemma}\label{lem: state}
	Under Assumptions \ref{ass: Lipschitz} and \ref{ass: terminal}, if \textbf{OCP$_i$} is feasible at $k^t_i$, the resulting state from the candidate control input sequence (\ref{eqn: feasible solution}) satisfies $\|\bar e_i(N|k^{t+1}_i)\|^2_{P_i} \leq \varepsilon^2_{i}$ if the following conditions hold:
	\begin{align}
		& \eta_i \leq \frac{L_{g_i}(\varepsilon_{ri} - \varepsilon_i) }{\lambda_{\text{max}}(\sqrt{P_i})\left[(1+L_{g_i})^N-1\right]}, \\
		& \Upsilon_i(m^t_i) \leq 0, \forall m^t_i \in \mathbb{N}_{[1,N]},
	\end{align}
	where
	\begin{eqnarray*}
		\Upsilon_i(m) &\hspace{-0.6em}=\hspace{-0.6em}& \left[1-\frac{\lambda_{\text{min}}(Q^*_i)}{\lambda_{\text{max}}(P_i)}\right]^m \Big[\varepsilon_i + \frac{\eta_i \lambda_{\text{max}}(\sqrt{P_i})} {L_{g_i}} \nonumber\\
		&\hspace{-0.6em}\hspace{-0.6em}& \left[(1+L_{g_i})^N - (1+L_{g_i})^{N-m}\right] \Big]^2 - \varepsilon_i^2.
	\end{eqnarray*}
\end{lemma}
\begin{proof}
	If $\tau \in \mathbb{N}_{[1, N-m^t_i]}$, under the same control sequence $\{\tilde u^{\star}_i(m^t_i|k^t_i),...,\tilde u^{\star}_i(N-1|k^t_i)\}$ and using Gronwall lemma \cite{xie2021robust}, we have
	\begin{align}\label{eqn: state error, prior interval}
		&\|\bar e_i(\tau|k^{t+1}_i) - \tilde e^{\star}_i(\tau+m^t_i|k^{t+1}_i) \|_{P_i} \nonumber\\
		&\leq \frac{\eta_i \lambda_{\text{max}}(\sqrt{P_i}) }{L_{g_i}} \left[(1+L_{g_i})^{m^t_i}-1\right](1+L_{g_i})^{\tau},
	\end{align}
	let $\tau = N-m^t_i$, it yields
	\begin{align*}
		&\|\bar e_i(N-m^t_i|k^{t+1}_i)\|_{P_i} \leq \|\tilde e^{\star}_i(N|k^{t+1}_i)\|_{P_i} \\
		&+ \frac{\eta_i \lambda_{\text{max}}(\sqrt{P_i}) }{L_{g_i}} \left[(1+L_{g_i})^{N} - (1+L_{g_i})^{N-{m^t_i}}\right] \\
		&\leq \varepsilon_i + (\varepsilon_{ri} - \varepsilon_i) \frac{(1+L_{g_i})^{N} - (1+L_{g_i})^{N-{m^t_i}}}{(1+L_{g_i})^{N}-1 } \\
		&\leq \varepsilon_i + \varepsilon_{ri} - \varepsilon_i = \varepsilon_{ri}.
	\end{align*}
	
	For $\tau \in \mathbb{N}_{[N-m^t_i+1, N]}$, the controller is switched into the  local auxiliary controller $K_i \bar e_i(\tau|k^{t+1}_i)$, Assumption \ref{ass: terminal} implies $\|\bar e_i(N|k^{t+1}_i)\|^2_{P_i} - \|\bar e_i(N-1|k^{t+1}_i)\|^2_{P_i} \leq -\|\bar e_i(N-1|k^{t+1}_i)\|^2_{Q^*_i}$, then it follows that $\|\bar e_i(N|k^{t+1}_i)\|^2_{P_i} \leq \left[1-\lambda_{\text{min}}(Q^*_i)/\lambda_{\text{max}}(P_i)\right] \|\bar e_i(N-1|k^{t+1}_i)\|^2_{P_i}$. Through iteration and $\Upsilon(m)\leq0$, we obtain $\|\bar e_i(N|k^{t+1}_i)\|^2_{P_i} \leq\left[1-\lambda_{\text{min}}(Q^*_i)/\lambda_{\text{max}}(P_i)\right]^m \|\bar e_i(N-m^t_i|k^{t+1}_i)\|^2_{P_i}\leq \varepsilon_i^2$. This completes the proof.
\end{proof}

The step length of the open-loop phase $m^t_i$ is determined by the following self-triggering theorem, and the next triggering moment can be decided as $k^{t+1}_i = k^t_i + m^t_i$.

\begin{theorem}\label{theo: trigger}
	At sampling instant $k^t_i$ for agent $i$, suppose Lemma \ref{lem: state} holds. The next triggering instant is calculated as $k^{t+1}_i = k^t_i + m^t_i$, where the open-loop phase $m^t_i$ is determined by the following condition:
	\begin{eqnarray}\label{eqn: trigger condition}
		m^t_i = \mathcal{I}_i + (1-\mathcal{I}_i) \min \{ \bar m_f, \bar m_s \},	
	\end{eqnarray}
where $\mathcal{I}_i = \mathcal{I} \{\psi_i(1)  \eta_i \leq \sigma_{i} [\|e_i(k^t_i)\|^2_{Q_i} + \|\tilde u^{\star}_i(0|k^t_i)\|^2_{R_i}]\}$, $\bar m_f = \sup_{m\in\mathbb{N}_{[2,N]}} \Upsilon_i(m) \leq 0$, $\bar m_s = \sup_{m\in\mathbb{N}_{[2,N]}} \psi_i(m) \eta_i \leq \sigma_i [\lambda_{\text{min}}(Q_i) \alpha_i(m)^2 + \|\tilde u^{\star}_i(m-1|k)\|^2_{R_i}]$, $\sigma_{i} \in (0,1)$ is called the triggering performance factor, $\alpha_i(m) = \max \{0, \|e_i(m-1|k)\| - m (1+L_{g_i})^{m-1} \eta_i\}$, $\psi_i(m) = \varpi_i(m) L_{Q_i} + \nu_i(m) \tau_i(m) L_{Q^*_i} + \nu_i(m) \varsigma_i(m) L_{P_i}$,  $\varpi_i(m) = (\nu_i(m)-1) / L_{g_i}$, $\nu_i(m) = (1+L_{g_i})^{N-m}$, $\tau_i(m) = (\varsigma_i(m)-1) / L_{\kappa_i}$, $\varsigma_i(m) = (1+L_{\kappa_i})^{m-1}$.

\end{theorem}
\begin{proof}
	Suppose the current sampling instant $k^t_i$, abbreviated as $k$ in the proof for brevity, deduces the optimal control input $\tilde u^{\star}_i(\tau|k)$ by minimizing $J_i(k)$ as in (\ref{OCP: cost function}). The Lyapunov function candidate is chosen as $V_i(e_i(k), \tilde{\textbf{u}}_i(k)) \triangleq H_i(e_i(k), \tilde{\textbf{u}}_i(k))$. First, we evaluate the difference of the Lyapunov function between instants $k$ and $k+1$ (i.e., $m^t_i=1$):	
	\begin{eqnarray}
		&\hspace{-0.6em}\hspace{-0.6em}& V_i( e_i(k+1), \bar{\textbf{u}}_i(k+1)) - V_i(e_i(k), \tilde{\textbf{u}}^{\star}_i(k)) \nonumber\\
		&\hspace{-0.6em}\hspace{-0.6em}& \leq -\|e_i(k)\|^2_{Q_i} - \|\tilde u^{\star}_i(0|k)\|^2_{R_i} + \Theta_1 + \Theta_2, \label{eqn: diff H=1}
	\end{eqnarray}
	where $\Theta_1 = \sum_{\tau=0}^{N-2} \|\bar e_i(\tau|k+1) \|^2_{Q_i} - \|\tilde e^{\star}_i(\tau+1|k) \|^2_{Q_i}$, $\Theta_2 = \|\bar e_i(N-1|k+1) \|^2_{Q_i} + \|\bar u_i(N-1|k+1) \|^2_{R_i} + \|\bar e_i(N|k+1) \|^2_{P_i} - \|\tilde e^{\star}_i(N|k) \|^2_{P_i}$. According to \cite{paula2006analysis} and \cite{wang2024rolling}, there always
	exist a constant $L_{Q_i}$ such that $\|e_1\|^2_{Q_i} - \|e_2\|^2_{Q_i} \leq L_{Q_i} \|e_1 - e_2\|_{Q_i}$ and a constant $L_{P_i}$ analogously, for $\forall e_1, e_2 \in \mathbb{E}_i$. Based on Gronwall Lemma, we have the following relationships:
	\begin{eqnarray}
		\Theta_1
		&\hspace{-0.6em}\leq\hspace{-0.6em}& L_{Q_i} \sum_{\tau=0}^{N-2} \|\bar e_i(\tau|k+1) - \tilde e^{\star}_i(\tau+1|k) \| \nonumber\\
		&\hspace{-0.6em}\leq\hspace{-0.6em}& L_{Q_i} \sum_{\tau=0}^{N-2} \Big[ \|\bar e_i(k+1) - \tilde e^{\star}_i(1|k)\| + \sum_{s=0}^{\tau-1} \| g_i(\bar e_i(s|k \nonumber\\
		&\hspace{-0.6em}\hspace{-0.6em}& +1), \tilde u^{\star}_i(s+1|k)) - g_i(\tilde e^{\star}_i(s+1|k),\tilde u^{\star}_i(s+1|k)) \|\Big] \nonumber\\
		&\hspace{-0.6em}\leq\hspace{-0.6em}& L_{Q_i} \sum_{\tau=0}^{N-2} (1+L_{g_i})^\tau \eta_i = \varpi_i(1) L_{Q_i} \eta_i. \label{eqn: Theta1}
	\end{eqnarray}
	
	Recall that Lemma \ref{lem: state} implies $\bar e_i(N-m^t_i|k^{t}_i + m^t_i) \in \Omega_i(\varepsilon_{ri})$ for $\forall m^t_i \in \mathbb{N}_{[1,N]}$, and by virtue of Assumption \ref{ass: terminal}, we have
	\begin{eqnarray}
		\Theta_2
		&\hspace{-0.6em}\leq\hspace{-0.6em}& \|\bar e_i(N-1|k+1) \|^2_{P_i} - \|\tilde e^{\star}_i(N|k) \|^2_{P_i} \nonumber\\
		&\hspace{-0.6em}\leq\hspace{-0.6em}& L_{P_i} \|\bar e_i(N-1|k+1) - \|\tilde e^{\star}_i(N|k) \|_{P_i} \nonumber\\
		&\hspace{-0.6em}\leq\hspace{-0.6em}& \nu_i(1) L_{P_i} \eta_i. \label{eqn: Theta2}
	\end{eqnarray}
	
	By substituting (\ref{eqn: Theta1}) and (\ref{eqn: Theta2}) into (\ref{eqn: diff H=1}) and using the triggering condition (\ref{theo: trigger}), one has
	\begin{eqnarray}
		&\hspace{-0.6em}\hspace{-0.6em}& V_i( e_i(k+1), \bar{\textbf{u}}_i(k+1)) - V_i(e_i(k), \tilde{\textbf{u}}^{\star}_i(k)) \nonumber\\
		&\hspace{-0.6em}\hspace{-0.6em}& \leq -\|e_i(k)\|^2_{Q_i} - \|\tilde u^{\star}_i(0|k)\|^2_{R_i}) + \psi_i(1)\eta_i \nonumber\\
		&\hspace{-0.6em}\hspace{-0.6em}& \leq (\sigma_{i}-1) (\|e_i(k)\|^2_{Q_i} + \|\tilde u^{\star}_i(0|k)\|^2_{R_i}) < 0, \label{eqn: diff H=1 ST}
	\end{eqnarray}
	which implies the convergence property of the stage cost and the optimal control problem (\ref{OCP: cost function}) is determined not to be solved at instant $k+1$. Then, we proceed to check the cost difference between instant  $k+\underline{m}\;(\underline{m} = m-1)$ and $k+m$, for $\forall m \geq 2$:
	\begin{align}
		& V_i(e_i(k+m),\bar{\textbf{u}}_i(k+m)) - V_i(e_i(k+\underline{m}),\bar{\textbf{u}}_i(k+\underline{m})) \nonumber\\
		&\leq -\|e_i(k+\underline{m})\|^2_{Q_i} \!-\! \|\tilde u^{\star}_i(\underline{m}|k)\|^2_{R_i} + \Delta_1 + \Delta_2 + \Delta_3, \label{eqn: diff H>1}
	\end{align}
	with $\Delta_1 = \sum_{\tau=0}^{N-m-1} \|\bar e_i(\tau|k+m) \|^2_{Q_i} - \|\bar e_i(\tau+1|k+\underline{m}) \|^2_{Q_i}$, $\Delta_2 = \sum_{\tau=N-m}^{N-2} \|\bar e_i(\tau|k+m) \|^2_{Q^*_i} - \|\bar e_i(\tau+1|k+\underline{m}) \|^2_{Q^*_i}$, $\Delta_3 = \|\bar e_i(N-1|k+m) \|^2_{Q_i} + \|\bar u_i(N-1|k+m) \|^2_{R_i} + \|\bar e_i(N|k+m) \|^2_{P_i} - \|\bar e_i(N|k+\underline{m}) \|^2_{P_i}$. It holds that
	\begin{align}
		\Delta_1
		&\leq  L_{Q_i} \sum_{\tau=0}^{N-m-1} \|\bar e_i(\tau|k+m) - \bar e_i(\tau+1|k+\underline{m}) \| \nonumber\\
		&\leq L_{Q_i} \sum_{\tau=0}^{N-m-1} \Big[ \|\bar e_i(k+m) - \bar e_i(1|k+\underline{m})\| \nonumber\\
		&\quad + \sum_{s=0}^{\tau-1} \| g_i(\bar e_i(s|k+m), \tilde u^{\star}_i(s+m|k)) \nonumber\\
		&\quad - g_i(\bar e_i(s+1|k+\underline{m}), \tilde u^{\star}_i(s+1+\underline{m}|k)) \|\Big] \nonumber\\
		&\leq L_{Q_i} \sum_{\tau=0}^{N-m-1} (1+L_{g_i})^\tau \eta_i = \varpi_i(m) L_{Q_i} \eta_i. \label{eqn: Delta1}
	\end{align}
	
	Specifically, we get $\|\bar e_i(N-m|k+m) - \bar e_i(N-m|k+\underline{m}) \| \leq \nu_i(m )\eta_i$, then $\Delta_2$ can be rewritten as
	\begin{align}
		\Delta_2
		& \leq L_{Q^*_i} \sum_{\tau=N-m}^{N-2} \|\bar e_i(\tau|k+m) - \bar e_i(\tau+1|k+\underline{m}) \| \nonumber\\
		& = L_{Q^*_i} \sum_{\tau=0}^{m-2} \|\bar e_i(\tau-m+N|k+m) \nonumber\\
		&\qquad\qquad\quad\;\; - \bar e_i(\tau+1-m+N|k+\underline{m}) \| \nonumber\\
		& \leq L_{Q^*_i} \nu_i(m) \sum_{\tau=0}^{m-2}  (1+L_{\kappa_i})^\tau \eta_i = \nu_i(m) \tau_i(m) L_{Q^*_i} \eta_i. \label{eqn: Delta2}
	\end{align}
	
	Recall Assumption (\ref{ass: terminal}), one has
	\begin{align}
		\Delta_3
		& \leq L_{P_i} \|\bar e_i(N-1|k+m) - \bar e_i(N|k+\underline{m}) \| \nonumber\\
		& \leq \nu_i(m) \varsigma_i(m) L_{P_i} \eta_i. \label{eqn: Delta3}
	\end{align}
	
	Nonetheless, the actual state $e_i(k+\underline{m})$ is unknown at instant $k$, since we have $\|e_i(\underline{m}|k)\|-\|e_i(k+\underline{m})\| \leq \|e_i(\underline{m}|k)-e_i(k+\underline{m})\| \leq m (1+L_{g_i})^{\underline{m}} \eta_i$. It follows that $-\|e_i(k+\underline{m})\| \leq -\alpha_i(m)$. Recall the triggering condition $\bar m_s$, combining (\ref{eqn: Delta1}), (\ref{eqn: Delta2}), (\ref{eqn: Delta3}) with (\ref{eqn: diff H>1}) yields that
	\begin{align}
		& V_i(e_i(k+m),\bar{\textbf{u}}_i(k+m)) - V_i(e_i(k+\underline{m}),\bar{\textbf{u}}_i(k+\underline{m})) \nonumber\\
		& \leq -\|e_i(k+\underline{m})\|^2_{Q_i} - \|\tilde u^{\star}_i(\underline{m}|k)\|^2_{R_i} + \psi_i(m) \eta_i \nonumber\\
		&\leq -\lambda_{\text{min}}(Q_i) \alpha_i(m)^2 - \|\tilde u^{\star}_i(\underline{m}|k)\|^2_{R_i} + \psi_i(m) \eta_i \nonumber\\
		&\leq (\sigma_i-1) [\lambda_{\text{min}}(Q_i) \alpha_i(m)^2 + \|\tilde u^{\star}_i(\underline{m}|k)\|^2_{R_i}] < 0. \label{eqn: diff H>1 ST}
	\end{align}
	
	The maximal admissible open-loop chase preserving convergence can be obtained.
	This completes the proof.
\end{proof}

Summarizing the above analysis, the pseudocode of the self-triggering cooperation dual-model strategy is generalized as in Algorithm 1.
\begin{algorithm}
	\caption{Self-Triggered Parameterized DMPC}\label{alg1}
	\begin{algorithmic}[1]
			
			\State
			\textbf{Initialization.} For each agent $i \in \mathcal{M}$, the prediction horizon $N$, the weighted matrix $Q_i$, $R_i$, $P_i$ the weighted factor $\rho_{ij}$, the auxiliary gain $K_i$, the justified terminal region parameters $\varepsilon_{ri}$, $\varepsilon_i$, the triggering factor $\sigma_i$.
			
			
			\State
			At sampling instant $k^t_i$, agent $i$ samples the actual error state $e_{i}(k^t_i)$ and parameter $s_i(k^t_i)$.
			
			\If{$e_{i}(k^t_i) \in \Omega_{r_i}$}
			\State
			Apply the auxiliary control law $u_{i}(k^t_i) = K_i e_{i}(k^t_i)$.
			
			\Else
			\State
			Construct the predicted synchronization parameter sequence $ \hat{\textbf{s}}_{i,j}(k^t_i)$ based on the most recent received values $\textbf{s}^{\star}_j(k^t_{i,j})$.
		    Solve the OCP$_i$ to obtain $\tilde{\mathbf{u}}^{\star}_{i}(k^t_i)$ and $\mathbf{s}^{\star}_{i}(k^t_i)$;
			
			\State
			Calculate the open-loop phase $m^t_i$ and the next triggering instant $m^{t+1}_i = k^t_i + m^t_i$;
			
			\State
			Broadcast $\textbf{s}^{\star}_i(k^t_i)$ to subsystem $j \in \mathcal{N}_i$;
			
			\State
			Apply the first $m^t_i$  control action in $\tilde{\mathbf{u}}^{\star}_{i}(k^t_i)$;
			
			\EndIf
			\State
			Update the sampling instant $k^{t+1}_i \to  k^t_i$, and go to Step 2 if the formation keeps moving on.
			
		\end{algorithmic}
\end{algorithm}

\begin{remark}
	It is worth mentioning that the proposed self-triggering asynchronous generator (\ref{eqn: trigger condition}) differs from the existing strategy, for examples, in \cite{chen2023asynchronous, mi2019self, wei2021self}. In our work, the neighbor information is not considered in the design of the trigger condition. Although the constraint (\ref{OCP: stability}) is introduced to guarantee the stability of the closed-loop system, the conservativeness and autonomy of the trigger policy can be efficiently reduced.
\end{remark}

\section{Feasibility and Stability Analysis}
In this section, the recursive feasibility of the proposed Algorithm 1 and the stability of closed-loop system are analyzed successively. The following theorem illustrates the recursive feasibility under some mild conditions.
\begin{theorem}\label{theo: feasibility}
	Suppose that Lemma \ref{lem: state} is satisfied and the open-loop phase for each agent $i$ is determined as in (\ref{eqn: trigger condition}), then \textbf{OCP$_i$} is recursively feasible if the \textbf{OCP$_i$} is feasible at the initial time $k^0_i$ and $\Omega(\varepsilon_{ri}) \subseteq \mathbb{E}_i \ominus \Lambda_i(N)$ under Algorithm 1.
\end{theorem}
\begin{proof}
	The control constraint (\ref{OCP: control}) trivially holds based on the assumption of feasibility at instant $k^t_i$ and Assumption \ref{ass: terminal}. With regard to the state constraint (\ref{OCP: state}), for the interval $\tau \in \mathbb{N}_{[1, N-m^t_i]}$, from the relationship (\ref{eqn: state error, prior interval}), we have $\bar e_i(\tau|k^{t+1}_i) \in \mathbb{E}_i \ominus \Lambda_i(\tau+m^t_i) \oplus \frac{\eta_i \lambda_{\text{max}}(\sqrt{P_i}) }{L_{g_i}} \left[(1+L_{g_i})^{m^t_i}-1\right](1+L_{g_i})^{\tau}  \in \mathbb{E}_i \ominus \Lambda_i(\tau)$; for $\tau \in \mathbb{N}_{[N-m^t_i+1, N]}$, the error state is forced into the terminal region $\Omega(\varepsilon_{ri})$, $\bar e_i(\tau|k^{t+1}_i) \in \Omega(\varepsilon_{ri}) \subseteq \mathbb{E}_i \ominus \Lambda_i(N) \subseteq \mathbb{E}_i \ominus \Lambda_i(\tau)$ can be clearly obtained. Lemma \ref{lem: state} shows the satisfaction of terminal constraint (\ref{OCP: terminal}). Finally, we turn to the proof of the constraint (\ref{OCP: stability}). Replacing $k^t_i + \underline{m}$ of (\ref{eqn: diff H>1 ST}) with $k^{t+1}-1$ yields
	\begin{align*}
		& H_i(e_i(k^{t+1}_i), \bar{\textbf{u}}_i(k^{t+1}_i)) \\
		& \leq -\|e_i(k^{t+1}_i-1)\|^2_{Q_i} - \|\tilde u^{\star}_i(m^t_i-1|k^t_i)\|^2_{R_i} \\
		&\quad + H_i(e_i(k^{t+1}_i-1), \bar{\textbf{u}}_i(k^{t+1}_i-1))  + \psi_i(m^t_i)\eta_i \\
		& = \phi_i(k^{t+1}_i),
	\end{align*}
	which implies the satisfaction of the constraint (\ref{OCP: stability}). This completes the proof.
\end{proof}

The core is to prove that the candidate control sequence (\ref{eqn: feasible solution}) meets constraints (\ref{OCP: control})-(\ref{OCP: stability}).
We also give the following theorem concerning the stability.

\begin{theorem}\label{theo: stability}
	Suppose that Theorem \ref{theo: trigger} and Theorem \ref{theo: feasibility} hold. For any initial state belongs to the set $\mathbb{E}_i\setminus\Omega_i(\varepsilon_{ri})$, the $i$-th closed-loop system is ISS and the state converges into $\Omega_i(\varepsilon_{ri})$ within a  finite time under the proposed Algorithm 1.
\end{theorem}
\begin{proof}
	The subsystem $i$ is ISS if it admits an ISS-Lyapunov function, refer to \cite{jiang2001input}. As shown in Theorem \ref{theo: trigger}, the stage cost $H_i$ is selected as a Lyapunov function candidate. It is obvious that this function satisfies the prior condition  of Definition \ref{def: ISS}. The second condition on Lyapunov function decreasing is discussed in the following four cases:	
	\begin{itemize}
		\item[i)] Instant $k+1$ is not triggered and instant $k$ is triggered. From the eqnarray (\ref{eqn: diff H=1 ST}) based on the triggering criteria, the following inequality is satisfied:
		\begin{align*}
			& V_i(e_i(k+1),\bar{\textbf{u}}_i(k+1)) - V_i(e_i(k), \tilde{\textbf{u}}^{\star}_i(k)) \\
			& \leq -\|e_i(k)\|^2_{Q_i} + \psi_i(1) \eta_i.
		\end{align*}
	\end{itemize}
	
	\begin{itemize}
		\item[ii)] Instant $k+1$ is not triggered and instant $k$ is not triggered. Then the latest updating instant is set as $[k^t_i] = \arg \min_{k^t_i \in \mathbb{N}_{[1,N]}} k-k^t_i$. From (\ref{eqn: diff H>1 ST}), we have
		\begin{align*}
			& V_i(e_i(k+1),\bar{\textbf{u}}_i(k+1)) - V_i(e_i(k),\bar{\textbf{u}}_i(k)) \\
			& \leq -\|e_i(k)\|^2_{Q_i} + \psi_i(k-[k^t_i]+1) \eta_i.
		\end{align*}
	\end{itemize}
	
	\begin{itemize}
		\item[iii)] Instant $k+1$ is triggered and instant $k$ is also triggered. From the stability constraint (\ref{OCP: stability}), the inequality concerning Lyapunov decreasing is the same as the first case.
	\end{itemize}
	
	\begin{itemize}
		\item[iv)] Instant $k+1$ is triggered and instant $k$ is not triggered. Assume that instant $k+1$ is the $t+1$-th triggering moment. Then the prior updating instant is $k^t_i = \arg \min_{k^t_i \in \mathbb{N}_{[2,N]}} k^{t+1}_i-k^t_i$. From (\ref{eqn: phi}), we have
		\begin{align*}
			& V_i(e_i(k+1), \tilde{\textbf{u}}^{\star}_i(k+1)) - V_i(e_i(k),\bar{\textbf{u}}_i(k)) \\
			& \leq -\|e_i(k)\|^2_{Q_i} + \psi_i(m^t_i) \eta_i.
		\end{align*}
	\end{itemize}
	
	In summary, there always exists a candidate Lyapunov function whose generation value is not greater than that of the previous step. Then, the subsystem $i$ is ISS and the error state trajectory will enter towards terminal region $\Omega_i(\varepsilon_{ri})$ in finite time by following the argument in \cite{michalska1993robust}. This completes the proof.
\end{proof}

\section{Illustrative Example}
In this section, we shall present an illustrative example to show the efficiency of the proposed Algorithm 1. Consider the formation application for three nonholonomic wheeled mobile robots \cite{qin2023event}, the nominal error model is discretized as
\begin{eqnarray*}
	x_{ei}(k+1) &\hspace{-0.6em}=\hspace{-0.6em}& x_{ei}(k) + T [w_i(k)y_{ei}(k) + v_{ei}(k) ], \\
	y_{ei}(k+1) &\hspace{-0.6em}=\hspace{-0.6em}& y_{ei}(k) + T [-w_i(k)x_{ei}(k) + v_{ri}(k) \sin \theta_{ei}(k) ], \\
	\theta_{ei}(k+1) &\hspace{-0.6em}=\hspace{-0.6em}& \theta_{ei}(k) + Tw_{ei}(k),
\end{eqnarray*}
where $x_{ei}(k)$, $y_{ei}(k)$ and $\theta_{ei}(k)$ are the lateral position error, longitudinal position error and angle error, respectively. The error state is selected as $e_i(k) = [x_{ei}(k)^T, y_{ei}(k)^T, \theta_{ei}(k)^T]^T$, and the control input is $u_i(k) = [v_{ei}(k)^T, w_{ei}(k)^T]^T$, where $v_{ei}(k) = v_{ri}(k) \cos \theta_{ei}(k) - v_i(k)$ and $v_{ei}(k)w_{ei}(k) = w_{ri}(k) - w_i(k)$ are the linear velocity error and angular velocity error, where $v_{ri}(k)$, $w_{ri}(k)$ are the reference linear velocity and angular velocity, $v_{i}(k)$, $w_{i}(k)$ are the actual ones. The reference path is $\Gamma_1 = [3\cos(0.1s_1(k))-2\cos(0.1s_2(k)), 3\sin(0.1s_1(k))-2\sin(0.1s_2(k))]$, $\Gamma_2 = [3\cos(0.1s_2(k)), 3\sin(0.1s_2(k))]$, $\Gamma_3 = [3\cos(0.1s_3(k))+2\cos(0.1s_2(k)+\pi/3), 3\sin(0.1s_3(k))+2\sin(0.1s_2(k)+\pi/3)]$. The relevant reference velocity can be obtained by differential flattening technique, and details refer to \cite{qin2020formation}. The optional matrices are $Y_i = [0, -1]$, $Z_i = [0, 1]$. The sampling period is set as $T=0.2s$. The subsystems are subject to the state error constraint $\mathbb{E}_i = \{ e_i : |x_{ei}|\leq 0.3, |y_{ei}|\leq 0.3, |\theta_{ei}|\leq \pi/10 \}$, the  control error constraint $\mathbb{U}_i = \{ u_i : |v_{ei}|\leq1, |w_{ei}|<1 \}$. Using our proposed Algorithm 1 and initializing parameter: $N=6$, $Q_i = 3I_3$, $R_i = 0.01I_2$, $\rho_{ij}=1$, $\varepsilon_i=0.05$, $\varepsilon_{ri}=0.06$. According to Lemmas 3.2 and 3.3 in \cite{KG2002Nonlinear}, the Lipschitz constants in Assumption \ref{ass: Lipschitz} are found as
\begin{eqnarray*}
	L_{g_1} = 0.24, \;	L_{g_2} = 0.28,	\; L_{g_3} = 0.3072, \\
	L_{\kappa_1} = 0.44, \; L_{\kappa_2} = 0.48, \; L_{\kappa_3} = 0.5072,
\end{eqnarray*}
the weighted matrices
\begin{eqnarray*}
	P_1 &\hspace{-0.6em}=\hspace{-0.6em}&
	\left[\begin{matrix}
		   12.2730 & \;\;2.8905 & \;\;2.9541 \\
		\;\;2.8905 &    12.3029 & \;\;3.0116 \\
		\;\;2.9541 & \;\;3.0116 &    11.3057
	\end{matrix}\right], \\
	P_2 &\hspace{-0.6em}=\hspace{-0.6em}&
	\left[\begin{matrix}
		   12.2801 & \;\;2.9043 & \;\;2.9509 \\
		\;\;2.9043 &    12.3029 & \;\;2.9942 \\
		\;\;2.9509 & \;\;2.9942 &    11.3019
\end{matrix}\right], \\
	P_3 &\hspace{-0.6em}=\hspace{-0.6em}&
	\left[\begin{matrix}
		   12.2998 & \;\;2.8734 & \;\;3.0342 \\
		\;\;2.8734 &    12.3216 & \;\;3.0751 \\
		\;\;3.0342 & \;\;3.0751 &    11.3673
\end{matrix}\right],
\end{eqnarray*}
and the auxiliary controller gain matrices are
\begin{eqnarray*}
	K_1 &\hspace{-0.6em}=\hspace{-0.6em}&
	\left[\begin{matrix}
		-1.9457 & -1.9725 & -1.9827 \\
		-1.9484 & -1.9753 & -1.9858
	\end{matrix}\right], \\
	K_2 &\hspace{-0.6em}=\hspace{-0.6em}&
	\left[\begin{matrix}
		-1.9389 & -1.9657 & -2.0032 \\
		-1.9438 & -1.9709 & -2.0101
	\end{matrix}\right], \\
	K_3 &\hspace{-0.6em}=\hspace{-0.6em}&
	\left[\begin{matrix}
		-1.9358 & -1.9626 & -2.0198 \\
		-1.9423 & -1.9695 & -2.0299
	\end{matrix}\right].
\end{eqnarray*}

From Lemma \ref{lem: state}, the allowable disturbance bound is calculated as $\eta_1 \leq 1.3163\exp{-4}$, $\eta_2 \leq 1.1748\exp{-4}$, $\eta_3 \leq 1.0963\exp{-4}$, the parameter $\mu_i = 1.02$, and the triggering factors $\sigma_1 = 0.04$, $\sigma_2 = 0.02$, $\sigma_3 = 0.04$. The initial state error is $e_1(0) = [-0.2, 0.2, -0.1]^T$, $e_2(0) = [0.1, 0.2, -0.1]^T$, $e_3(0) = [0.2, -0.2, -0.1]^T$. The simulation results under Algorithm 1 are shown in Figs. \ref{Fig: state} - \ref{Fig: trigger}.  The state error responses and control trajectories are depicted in Fig. \ref{Fig: state} and Fig. \ref{Fig: control}, respectively. It follows from Fig. \ref{Fig: syn} that the synchronization parameters are converged to zero. According to Fig. \ref{Fig: trigger}, the proposed self-triggered strategy reduces the solving frequency efficiently.  The movement paths of three robots are shown in Fig. \ref{Fig: movement}, which illustrates that the robots are forced to follow the desired path with inter-robots formation coordination.

\begin{figure}
	\includegraphics[width=8.7cm]{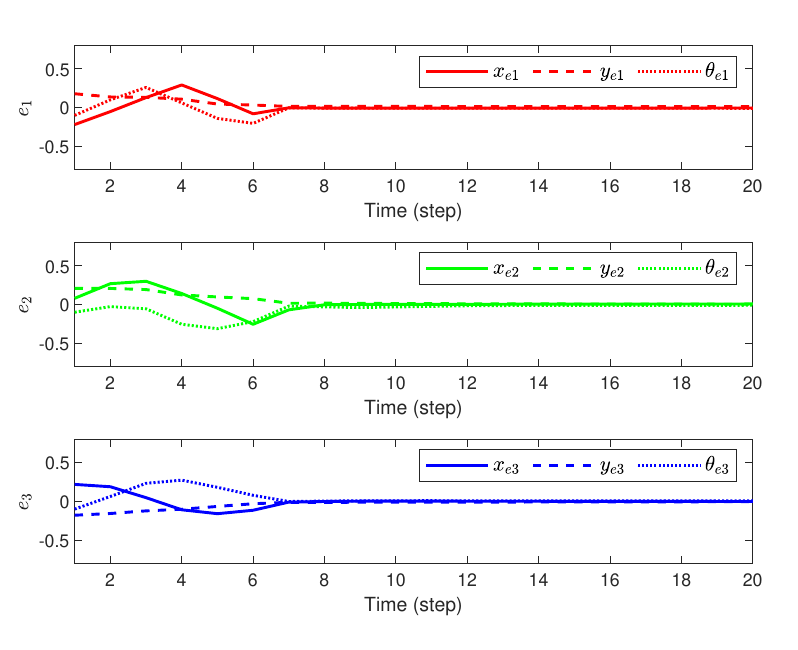}
	\caption{The state responses within 0-20 step.}
	\label{Fig: state}
\end{figure}

\begin{figure}
	\includegraphics[width=8.7cm]{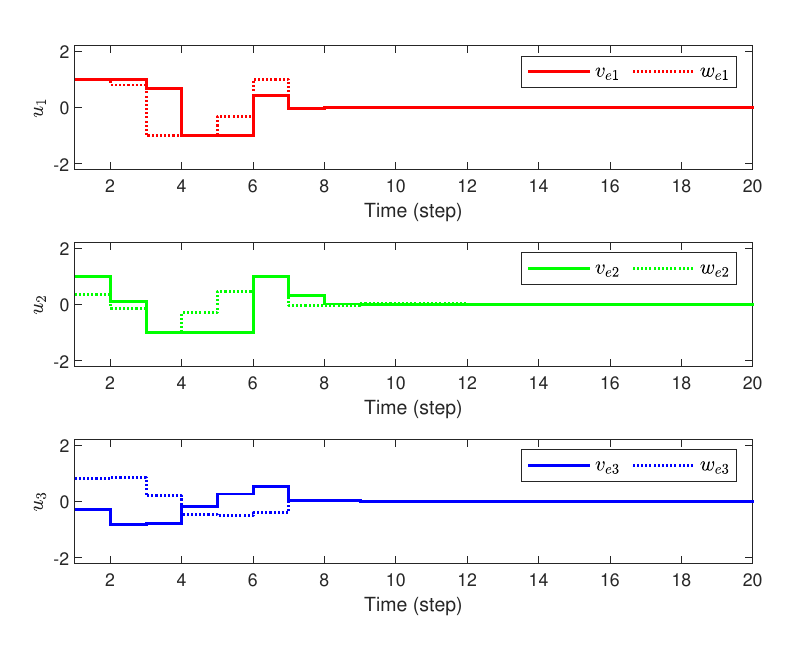}
	\caption{The control trajectories within 0-20 step.}
	\label{Fig: control}
\end{figure}

\begin{figure}
	\includegraphics[width=8.7cm]{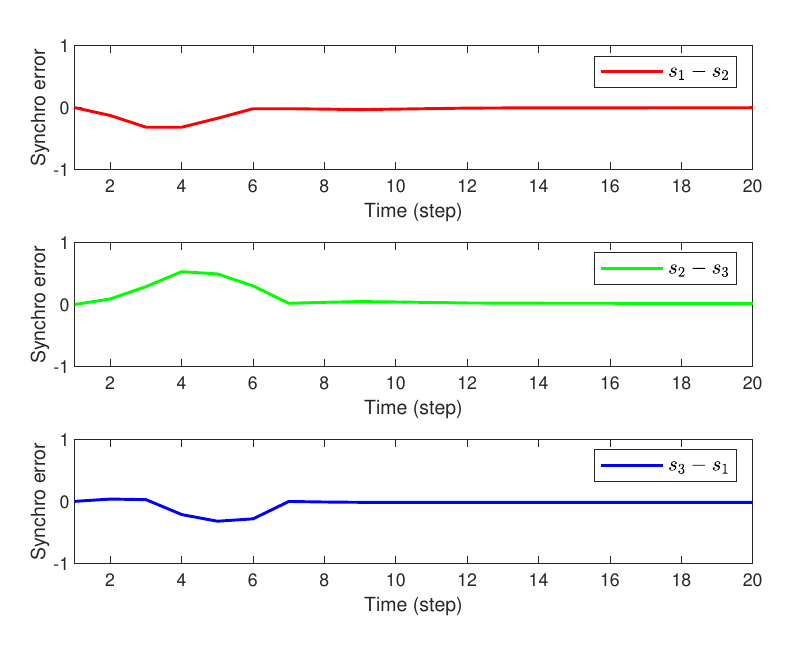}
	\caption{The error trajectories of the synchronization parameters within 0-20 step.}
	\label{Fig: syn}
\end{figure}

\begin{figure}
	\includegraphics[width=8.7cm]{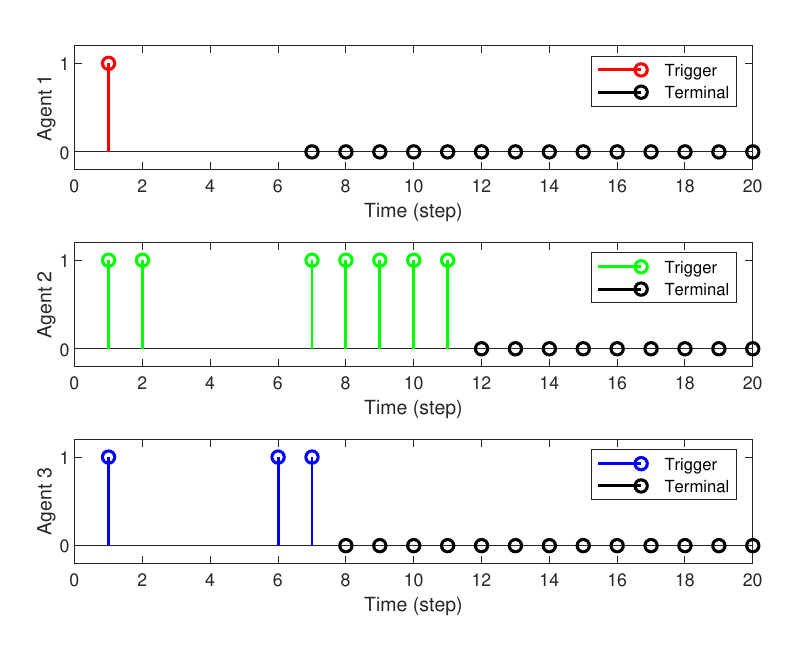}
	\caption{The triggering instants within 0-20 step.}
	\label{Fig: trigger}
\end{figure}

\begin{figure}
	\includegraphics[width=8.7cm]{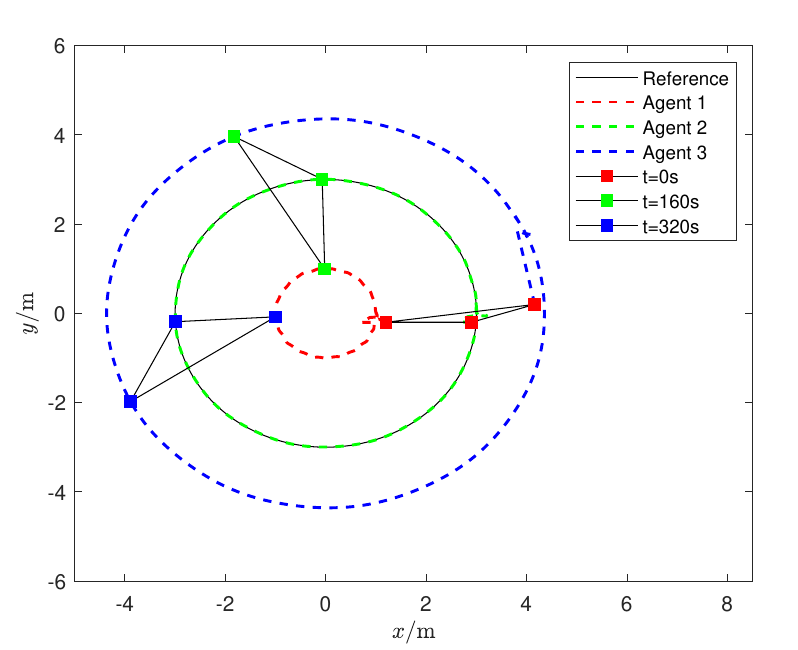}
	\caption{The movement paths of the robots.}
	\label{Fig: movement}
\end{figure}

\section{Conclusions}
This paper considers the problem of self-triggered DMPC approach for nonlinear MASs subject to bounded disturbances. The solving frequency is efficiently reduced through the novel self-triggered scheduler. On the other hand, the synchronization parameter constraints alleviate the interaction resource. By implementing the proposed algorithm, the balance of the performance level and resource usage for the closed-loop system are achieved. Rigorous analysis with respect to the recursive feasibility and stability is conducted. Finally, an illustrative example is presented to demonstrate the effectiveness of the proposed technique.

\bibliographystyle{ieeetr}
\bibliography{references}

\begin{thebibliography}{10}

\bibitem{yang2021distributed}
H.~Yang, T.~Li, Y.~Long, C.~P. Chen, and Y.~Xiao, ``Distributed virtual inertia
  implementation of multiple electric springs based on model predictive control
  in dc microgrids,'' {\em IEEE Transactions on Industrial Electronics},
  vol.~69, no.~12, pp.~13439--13450, 2021.

\bibitem{yin2020deep}
X.~Yin and X.~Zhao, ``Deep neural learning based distributed predictive control
  for offshore wind farm using high-fidelity les data,'' {\em IEEE Transactions
  on Industrial Electronics}, vol.~68, no.~4, pp.~3251--3261, 2020.

\bibitem{christofides2013distributed}
P.~D. Christofides, R.~Scattolini, D.~M. de~la Pena, and J.~Liu, ``Distributed
  model predictive control: A tutorial review and future research directions,''
  {\em Computers \& Chemical Engineering}, vol.~51, pp.~21--41, 2013.

\bibitem{zhou2022multiple}
X.~Zhou, T.~Yang, Y.~Zou, S.~Li, and H.~Fang, ``Multiple subformulae
  cooperative control for multiagent systems under conflicting signal temporal
  logic tasks,'' {\em IEEE Transactions on Industrial Electronics}, vol.~70,
  no.~9, pp.~9357--9367, 2022.

\bibitem{wang2023distributed}
J.~Wang and S.~Li, ``Distributed model predictive control for consensus of
  multi-agent systems with connectivity maintenance,'' {\em IEEE Transactions
  on Circuits and Systems I: Regular Papers}, 2023.

\bibitem{yuan2023distributed}
Q.~Yuan and X.~Li, ``Distributed model predictive formation control for a group
  of uavs with spatial kinematics and unidirectional data transmissions,'' {\em
  IEEE Transactions on Network Science and Engineering}, 2023.

\bibitem{hao2023trajectory}
L.~Hao, R.~Wang, C.~Shen, and Y.~Shi, ``Trajectory tracking control of
  autonomous underwater vehicles using improved tube-based model predictive
  control approach,'' {\em IEEE Transactions on Industrial Informatics}, 2023.

\bibitem{hou2021distributed}
B.~Hou, S.~Li, and Y.~Zheng, ``Distributed model predictive control for
  reconfigurable systems with network connection,'' {\em IEEE Transactions on
  Automation Science and Engineering}, vol.~19, no.~2, pp.~907--918, 2021.

\bibitem{zheng2023distributed}
Y.~Zheng, S.~Li, R.~Wan, Z.~Wu, and Y.~Zhang, ``Distributed model predictive
  control for reconfigurable systems based on lyapunov analysis,'' {\em Journal
  of Process Control}, vol.~123, pp.~1--11, 2023.

\bibitem{yang2021economic}
Y.~Yang, Y.~Zou, and S.~Li, ``Economic model predictive control of enhanced
  operation performance for industrial hierarchical systems,'' {\em IEEE
  Transactions on Industrial Electronics}, vol.~69, no.~6, pp.~6080--6089,
  2021.

\bibitem{mi2019self}
X.~Mi, Y.~Zou, S.~Li, and H.~R. Karimi, ``Self-triggered dmpc design for
  cooperative multiagent systems,'' {\em IEEE Transactions on Industrial
  Electronics}, vol.~67, no.~1, pp.~512--520, 2019.

\bibitem{chen2023asynchronous}
J.~Chen, H.~Wei, H.~Zhang, and Y.~Shi, ``Asynchronous self-triggered stochastic
  distributed mpc for cooperative vehicle platooning over vehicular ad-hoc
  networks,'' {\em IEEE Transactions on Vehicular Technology}, 2023.

\bibitem{wei2021self}
H.~Wei, K.~Zhang, and Y.~Shi, ``Self-triggered min--max dmpc for asynchronous
  multiagent systems with communication delays,'' {\em IEEE Transactions on
  Industrial Informatics}, vol.~18, no.~10, pp.~6809--6817, 2021.

\bibitem{wang2024rolling}
T.~Wang, Y.~Kang, P.~Li, Y.~Zhao, and H.~Tang, ``Rolling self-triggered
  distributed mpc for dynamically coupled nonlinear systems,'' {\em
  Automatica}, vol.~160, p.~111444, 2024.

\bibitem{skjetne2004robust}
R.~Skjetne, T.~I. Fossen, and P.~V. Kokotovi{\'c}, ``Robust output maneuvering
  for a class of nonlinear systems,'' {\em Automatica}, vol.~40, no.~3,
  pp.~373--383, 2004.

\bibitem{zhang2012distributed}
Q.~Zhang, L.~Lapierre, and X.~Xiang, ``Distributed control of coordinated path
  tracking for networked nonholonomic mobile vehicles,'' {\em IEEE Transactions
  on Industrial Informatics}, vol.~9, no.~1, pp.~472--484, 2012.

\bibitem{qin2023event}
D.~Qin, Z.~Jin, A.~Liu, W.-a. Zhang, and L.~Yu, ``Event-triggered distributed
  predictive cooperation control for multi-agent systems subject to bounded
  disturbances,'' {\em Automatica}, vol.~157, p.~111230, 2023.

\bibitem{qin2023asynchronous}
D.~Qin, Z.~Jin, A.~Liu, W.-A. Zhang, and L.~Yu, ``Asynchronous event-triggered
  distributed predictive control for multi-agent systems with parameterized
  synchronization constraints,'' {\em IEEE Transactions on Automatic Control},
  2023.

\bibitem{sun2019robust}
Z.~Sun, L.~Dai, K.~Liu, V.~Dimarogonas, Dimos, and Y.~Xia, ``Robust
  self-triggered mpc with adaptive prediction horizon for perturbed nonlinear
  systems,'' {\em IEEE Transactions of Automatic Control}, vol.~11, no.~64,
  pp.~4780--4787, 2019.

\bibitem{xie2021robust}
H.~Xie, L.~Dai, Y.~Luo, and Y.~Xia, ``Robust mpc for disturbed nonlinear
  discrete-time systems via a composite self-triggered scheme,'' {\em
  Automatica}, vol.~127, p.~109499, 2021.

\bibitem{paula2006analysis}
R.~Paulavičius and J.~Žilinskas, ``Analysis of different norms and
  corresponding lipschitz constants for global optimization,'' {\em
  Technological and Economic Development of Economy}, vol.~12, no.~4,
  pp.~301--306, 2006.

\bibitem{jiang2001input}
Z.~Jiang and Y.~Wang, ``Input-to-state stability for discrete-time nonlinear
  systems,'' {\em Automatica}, vol.~37, no.~6, pp.~857--869, 2001.

\bibitem{michalska1993robust}
H.~Michalska and D.~Q. Mayne, ``Robust receding horizon control of constrained
  nonlinear systems,'' {\em IEEE Transactions on Automatic Control}, vol.~38,
  no.~11, pp.~1623--1633, 1993.

\bibitem{qin2020formation}
D.~Qin, A.~Liu, D.~Zhang, and H.~Ni, ``Formation control of mobile robot
  systems incorporating primal-dual neural network and distributed predictive
  approach,'' {\em Journal of the Franklin Institute}, vol.~357, no.~17,
  pp.~12454--12472, 2020.

\bibitem{KG2002Nonlinear}
H.~K. Khalil and J.~W. Grizzle, {\em Nonlinear Systems}.
\newblock Prentice hall, 2002.

\end{thebibliography}
\end{document}